\documentclass[letterpaper]{IEEEtran}
\usepackage[T1]{fontenc}
\usepackage[latin9]{inputenc}
\usepackage{float}
\usepackage{amsmath}
\usepackage{color}
\usepackage{graphicx}
\usepackage{amssymb}

\makeatletter

\newcommand{\noun}[1]{\textsc{#1}}
\floatstyle{ruled}
\newfloat{algorithm}{tbp}{loa}
\floatname{algorithm}{Algorithm}

\newtheorem{example}{Example}
\newtheorem{definitn}{Definition}
\newtheorem{remrk}{Remark}
\newtheorem{thm}{Theorem}

\usepackage{draftwatermark}
\SetWatermarkLightness{0.90}

\usepackage{bm}
\renewcommand{\mathbf}[1]{\bm{#1}}

\newcommand{\Midx}{m}			% index for moments
\newcommand{\Mrunidx}{l}			% run-index for moments
\newcommand{\Mdeptha}{i}			% depth 1 in trellis
\newcommand{\Mdepthb}{j}			% depth 2 in trellis
\newcommand{\MsymS}{x}		% symbol
\newcommand{\Mtrellis}[2]{\theta^{(#1)}(#2)}	% trellis moment
\newcommand{\MNtrellis}[2]	       {\bar{\theta}^{(#1)}(#2)}	% (normalized trellis) moment
\newcommand{\Msymbol}[3]{\Omega^{(#1)}_{#2}(#3)}	% symbol moment
\newcommand{\MNsymbol}[3]{\bar{\Omega}^{(#1)}_{#2}(#3)}	% normalized symbol moment

\newcommand{\TT}{T}			% trellis
\newcommand{\Trank}{n}			% rank of trellis (code length)
\newcommand{\PP}	{\mathtt{P}}	%{\mathcal{P}}		% path

\newcommand{\KI}      {K_{1}}		% constant No. 1
\newcommand{\KIa}    {K_{1a}}		% constant No. 1a
\newcommand{\KIb}    {K_{1b}}		% constant No. 1b
\newcommand{\KII}    {K_{2}}			% constant No. 2
\newcommand{\Dtrellis}[1]  {\theta^{\fD}(#1)}		% trellis distribution
\newcommand{\Dsymbol}[2]   {\Omega^{\fD}_{#1}(#2)}	% symbol distribution
\newcommand{\DsymbolN}[2]   {\bar{\Omega}^{\fD}_{#1}(#2)}	% normalized symbol distribution

\newcommand{\Nslots}	{N}		% number of partitions for quantization
\newcommand{\Swidth}	{\Delta_w}		% partition width
\newcommand{\DmuDelta}	{\Delta_{\mu}}		% difference of mean values
\newcommand{\Didx}	{j}		% partition index (for distribution implementation)
\newcommand{\Dfraclow}	{\epsilon}
\newcommand{\Dd}	{\mathbf{d}}		% vector d for discretized distribution
\newcommand{\Dmunew}	{\mu_{\mathrm{new}}}	
\newcommand{\Dmuin}	{\mu_{\mathrm{in}}}	
\newcommand{\Del}[2]	{\mathbf{d}_{#1}[#2]}		% vector d for discretized distribution

\newcommand{\vFlow}[2]	{\eta(#1,#2)}		% flow
\newcommand{\vSIdent}	{1^{\oDot}}		% identity element semi-ring
\newcommand{\vSNull}	{0^{\oPlus}}		% null-element semi-ring
\newcommand{\vu}	{q}		% the `u' =cr^T of Uli's paper
\newcommand{\vmmax}	{M}		% maximum moment to calculate

\newcommand{\sV}	{\mathbb{V}}		% set of vertices
\newcommand{\sE}	{\mathbb{E}}		% set of edges
\newcommand{\sS}	{\mathbb{S}}		% algebraic set S
\newcommand{\sC}	{\mathbb{C}}		% code, i.e., set of code words
\newcommand{\sD}	{\mathbb{D}}		% domain of hard decision distribution
\newcommand{\sR}	{\mathbb{R}}		% real valued numbers

\newcommand{\fE}	{\mathrm{E}}		% expectation operator
\newcommand{\fs}	{c}		% symbol label
\newcommand{\fD}	{\mathcal{D}}		% distribution

\newcommand{\oPlus}	{\oplus}		% semi-ring sum
\newcommand{\oDot}	{\odot}		% semi-ring product
\newcommand{\oSum}	{\mathop{\sum \hspace{-3.5mm} \oPlus} }
\newcommand{\oSumFoot}	{\mathop{\sum \hspace{-2.8mm} \oPlus \hspace{1mm} } }
\newcommand{\pe}	{\hspace{1mm}+\hspace{-1.3mm}=}

\DeclareMathSymbol{\minus}{\mathord}{operators}{"2D}

\makeatother

\begin{document}

\title{Trellis Computations}

\author{Axel Heim, Vladimir Sidorenko, Uli Sorger}

\maketitle
\begin{abstract}
For a certain class of functions, the distribution of the function
values can be calculated in the trellis or a sub-trellis. The forward/backward
recursion known from the BCJR algorithm \cite{BCJR1974} is generalized
to compute the moments of these distributions. In analogy to the
symbol probabilities, by introducing a constraint at a certain depth
in the trellis we obtain symbol moments. These moments are required
for an efficient implementation of the discriminated belief propagation
algorithm in \cite{Sorger2007}, and can furthermore be utilized to
compute conditional entropies in the trellis. 

The moment computation algorithm has the same asymptotic complexity
as the BCJR algorithm. It is applicable to any commutative semi-ring,
thus actually providing a generalization of the Viterbi algorithm
\cite{Viterbi1967}.

\end{abstract}

\section{Introduction}

Trellises were introduced into the coding theory literature by Forney
\cite{Forney1967} as a means of describing the Viterbi algorithm
for decoding convolutional codes. Bahl et al. \cite{BCJR1974} showed
that block codes can also be described by a trellis, and Wolf \cite{Wolf1978}
proposed the use of the Viterbi algorithm for trellis-based soft-decision
decoding of block codes. Massey \cite{Massey1978} gave a graph-theoretic
definition of a block trellis and an alternative construction of minimal
trellises. Forney's paper \cite{Forney1988} showed that group codes,
including linear codes and lattices, have a well-defined trellis structure.

In \cite{McEliece1996}, McEliece investigated the complexity of a
generalized Viterbi algorithm which allows efficient computation of
flows on a code trellis. These results were further generalized in
\cite{Aji2000} and \cite{Kschischang2001}. However, the calculation
of flows does not fully exploit the capabilities of the trellis (representation):
For a certain set of functions it is possible to calculate the moments
of these functions in the trellis. These can be scalar or vectorial,
as long as they are linear and fulfill a separability criterion.

For iterative decoding of coupled codes, the popular sum-product algorithm
is used to calculate the symbol probabilities of the component codes.
These probabilities are exchanged between component decoders until
a stable solution is found. This iterative algorithm works very well
for long {}``turbo'', low-density parity check (LDPC) and some other
codes, obtained by concatenation of simple component codes in a special
way. However, performance becomes poor when utilizing short or some
good component codes. 

Recently, Sorger \cite{Sorger2007} showed that iterative decoding
is improved when discriminating code words $\mathbf{c}$ by their
correlation $\mathbf{c}\mathbf{r}^{T}$ or $\mathbf{c}\mathbf{w}^{T}$
with the received word $\mathbf{r}$ or a `believed' word $\mathbf{w}$,
respectively. Not only symbol probabilities are considered, but also
the distribution of these probabilities over the correlation value.
An efficient algorithm is introduced using the first two moments to
approximate these distributions. 

In this paper we propose algorithms to compute both such distributions
and their moments in the trellis.

\begin{example}
\label{exa:SymbolDistribution}Consider Figure \ref{fig:Symbol-distribution}
\begin{figure}
\begin{centering}
\includegraphics{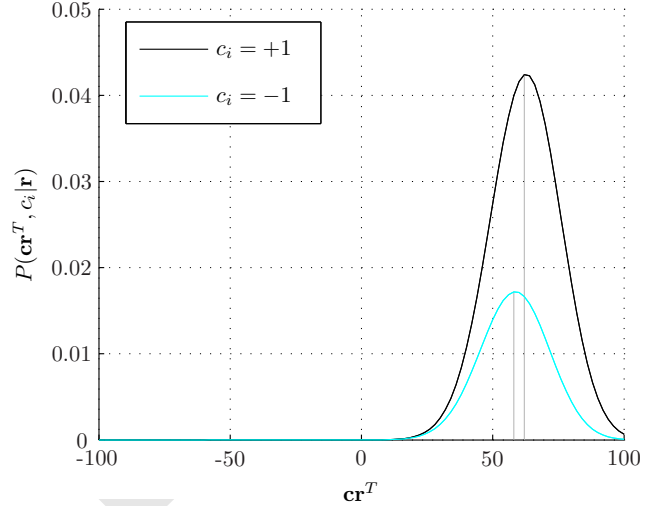}
\par\end{centering}

\caption{Symbol Distributions of Correlation $\mathbf{cr}^{T}$\label{fig:Symbol-distribution}}

\end{figure}
 which shows two distributions of the correlation function $\mathbf{c}\mathbf{r}^{T}$,
where $\mathbf{c}$ is a code word and $\mathbf{r}$ is the noisy
version of a code word $\check{\mathbf{c}}\in\sC$ after transmission
over a memory-less binary symmetric channel (BSC). The curves show
the distributions for $\mathbf{c}\in\sC_{i}(+1)$ and $\mathbf{c}\in\sC_{i}(-1)$,
respectively, where $\sC_{i}(x):=\left\{ \mathbf{c}\in\sC:c_{i}=x\right\} $
denotes the sub-code of $\sC$ for which the symbol $c_{i}$ at a
given position $i$ of each code word equals $\MsymS\in\{-1,+1\}$.
The integrals over the distributions equal the symbol probabilities
$P(c_{i}=x|\mathbf{r})$. However, the probability ratio\begin{equation}
\frac{P(\mathbf{c}\mathbf{r}^{T},c_{i}=+1|\mathbf{r})}{P(\mathbf{c}\mathbf{r}^{T},c_{i}=-1|\mathbf{r})}\label{eq:intro-ratio}\end{equation}
varies significantly over $\mathbf{cr}^{T}$ which can be exploited
when knowledge on the correlation $\check{\mathbf{c}}\mathbf{r}^{T}$
with the transmitted code word is available.

The distributions in Figure \ref{fig:Symbol-distribution} can be
approximated with their moments\begin{equation}
\fE_{\mathbf{C}}\left[\left.\left(\mathbf{cr}^{T}\right)^{\Midx}\right|\mathbf{r},c_{i}\right]:=\sum_{\mathbf{c}\in\sC}\left(\mathbf{cr}^{T}\right)^{\Midx}\cdot P(\mathbf{c}|\mathbf{r},c_{i})\label{eq:intro001}\end{equation}
up to a certain order $m$, where $\fE_{\mathbf{C}}[.]$ is the expectation
over all code words $\mathbf{c}\in\sC$. The distributions will be
\noun{Gauss}ian for sufficiently long codes which can be understood
by the law of large numbers. Hence we can expect the first two moments
to suffice for a good approximation.\medskip{}

\end{example}
We present generalizations of the methods in \cite{McEliece1996}
which enable us to compute distributions $P(\mathbf{c}\mathbf{w}^{T},c_{i}=+1|\mathbf{r})$
and expressions like $\fE_{\mathbf{C}}\left[\left(\mathbf{cw}^{T}\right)^{\Midx}|\mathbf{r},c_{i}\right]$
for some word $\mathbf{w}$, whereof (\ref{eq:intro001}) is a special
case, both for hard and soft decision. The complexity of the algorithm
is of the same order as the classically used BCJR algorithm.

The remainder of this paper is structured as follows. The next section
contains a review of common terminology in the context of trellises.
This is extended in Section \ref{sec:Trellis-Computations}, which
deals with the computation of distributions and their moments in a
more general frame. In Section \ref{sec:Applications} we will return
to the original problem by transferring the results of Section \ref{sec:Trellis-Computations}
to linear block codes and calculate the conditional entropy in the
trellis.

\section{\label{sec:Definitions}Definitions}

We deliberately follow to a wide extent the notation and style of
McEliece. The first paragraph is an excerpt from \cite{McEliece1996}
with minor modifications.%
\footnote{In contrast to \cite{McEliece1996} we restrict our definitions and
derivations to the set of real numbers.%
}

A \emph{trellis} $T=(\sV,\sE)$ of rank $n$ is a finite-directed
graph with vertex set $\sV$ and edge set $\sE$, in which every vertex
is assigned a \emph{depth} in the range $\{0,1,\dots,n\}$. Each edge
is connecting a vertex at depth $i-1$ to one at depth $i$, for some
$i\in\{0,1,\dots,\Trank\}$. Multiple edges between vertices are allowed.
The set of vertices at depth $i$ is denoted by $\sV_{i}$, so that
$\sV=\bigcup_{i=0}^{\Trank}\sV_{i}$. For $v\in\sV_{i}$ we write
$\mathrm{depth}(v)=i$. The set of edges connecting vertices at depth
$i-1$ to those at depth $i$ is denoted $\sE_{i-1,i},$ so that $\sE=\bigcup_{i=1}^{\Trank}\sE_{i-1,i}$.
There is only one vertex at depth $0$, called $A$, and only one
at depth $\Trank$, called $B$. If $e\in\sE$ is a directed edge
connecting the vertices $u$ and $v$, which we denote by $e:u\rightarrow v$,
we call $u$ the \emph{initial vertex}, and $v$ the \emph{final vertex}
of $e$ and write $\mathrm{init}(e)=u$, $\mathrm{fin(e})=v$. We
denote the number of edges leaving a vertex $v$ by $\rho^{+}(v)$,
and the number of edges entering a vertex $v$ by $\rho^{-}(v)$,
i.e.\begin{eqnarray*}
\rho^{+}(v) & = & |\{e:\mathrm{init}(e)=v\}|\\
\rho^{-}(v) & = & |\{e:\mathrm{fin}(e)=v\}|.\end{eqnarray*}
If $u$ and $v$ are vertices, a \emph{path} $\PP$ of \emph{length}
$L$ from $u$ to $v$ is a sequence of $L$ edges: $\PP=e_{1}e_{2}\cdots e_{L}$,
such that $\mathrm{init}(e_{1})=u$, $\mathrm{fin}(e_{L})=v$, and
$\mathrm{fin}(e_{i})=\mathrm{init}(e_{i+1})$, for $i=1,2,\dots,L-1$.
If $\PP$ is such a path, we sometimes write $\PP:u\rightarrow v$
for short, as well as $\mathrm{init}(\PP)=\mathrm{init}(e_{1})$ and
$\mathrm{fin}(\PP)=\mathrm{fin}(e_{L})$. We denote the set of paths
from vertices at depth $i$ to vertices at depth $j$ by $\sE_{i,j}$.
We assume that for every vertex $v\neq A,B$, there is at least one
path from $A$ to $v$, and at least one path from $v$ to $B$.

\begin{example}[Trellis]
\label{exa:Trellis}Figure \ref{fig:Trellis-SPC} %
\begin{figure}
\begin{centering}
\includegraphics[scale=0.75]{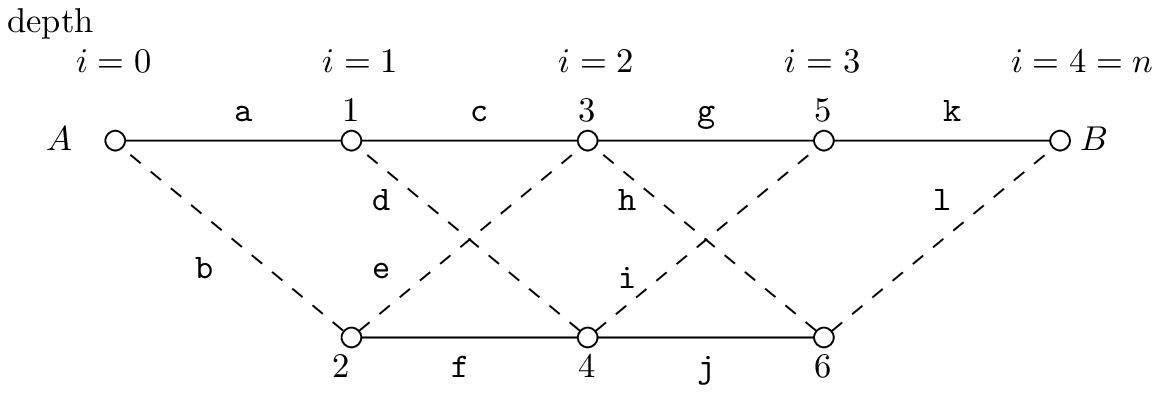}
\par\end{centering}

\caption{\label{fig:Trellis-SPC}Trellis of rank $\Trank=4$ with vertex set
$\sV=\{A,1,2,3,4,5,6,B\}$ and edge set $\sE=\{\mathtt{a},\mathtt{b},\mathtt{c},\mathtt{d},\mathtt{e},\mathtt{f},\mathtt{g},\mathtt{h},\mathtt{i},\mathtt{j},\mathtt{k},\mathtt{l}\}$}

\end{figure}
 shows a trellis of rank $\Trank=4$ with edge set $\sE=\{\mathtt{a},\mathtt{b},\mathtt{c},\mathtt{d},\mathtt{e},\mathtt{f},\mathtt{g},\mathtt{h},\mathtt{i},\mathtt{j},\mathtt{k},\mathtt{l}\}$
and vertex set $\sV=\{A,1,2,3,4,5,6,B\}$. There are eight paths $\PP:A\rightarrow B$
from $A$ to $B$. There is $\rho^{-}(1)=1$ edge entering (edge $\mathtt{a}$)
and $\rho^{+}(1)=2$ edges (edges $\mathtt{c}$ and $\mathtt{d}$)
leaving vertex $v=1$. 
\end{example}
We assume each edge in the trellis is \emph{labeled}. Let $T=(\sV,\sE)$
be a trellis of rank $\Trank$, such that each edge $e\in\sE$ is
labeled with a real valued number $\lambda(e)\in\sR$. We now define
the label of a path, and the flow between two vertices.

\begin{definitn}[Path Labels]
\label{def:label_of_path}The \emph{label} $\lambda(\PP)$ of a path
$\PP=e_{1}e_{2}\cdots e_{L}$ is defined as the product $\lambda(\PP)=\lambda(e_{1})\cdot\lambda(e_{2})\cdot\dots\,\cdot\lambda(e_{L})$
of the labels of all edges in the path. (Note that the subscript indicates
the sequence number rather than the edge's depth.)
\end{definitn}

\begin{definitn}[Flow]
\label{def:flow}If $u$ and $v$ are vertices in a labeled trellis,
we define the \emph{flow} $\vFlow{u}{v}$ from $u$ to $v$ to be
the sum of the labels on all paths from $u$ to $v$, i.e., \[
\vFlow{u}{v}=\sum_{\PP:u\rightarrow v}\lambda(\PP)\,.\]

\end{definitn}
In this paper, we only consider operations on the set of real numbers
with ordinary addition and multiplication as the authors are not aware
of application for other algebraic structures. However, Appendix \ref{sub:Appendix-Semi-Ring}
briefly shows that the algorithm can be transferred to any commutative
semi-ring, thus leading to a generalization of the Viterbi algorithm
\cite{Viterbi1967}.

\begin{example}
\label{exa:ChannelProbabilities}We continue Example \ref{exa:Trellis}.
The trellis depicted in Figure \ref{fig:Trellis-SPC} is the trellis
of the $(4,3,2)$ single parity check code. In the BCJR algorithm,
the edge labels $\lambda(e)$ are the channel probabilities of the
corresponding transitions.
\end{example}

\section{\label{sec:Trellis-Computations}Trellis-Based Computations}

In this section we consider distributions of the type \[
\mathcal{D}:\,\vu\mapsto D(\vu)=\sum_{{\PP:A\rightarrow B\atop f(\PP)=\vu}}\lambda\left(\PP\right)\]
 for special functions $f$, i.e., $\vu$ is mapped to the sum of
the labels of all paths $\PP$ with $f(\PP)=\vu$. We present an algorithm
to calculate these distributions over all paths of a trellis or a
sub-set of these. Before, however, we develop algorithms to calculate
the moments \[
\MNtrellis{\Midx}{\TT}:=\frac{\sum_{\PP}\left(f(\PP)\right)^{\Midx}\cdot\lambda(\PP)}{\sum_{\PP}\lambda(\PP)}\]
and - by introducing a constraint on the paths - the symbol moments
\[
\MNsymbol{\Midx}{i}{\TT,\MsymS}:=\frac{{\displaystyle \sum_{{\PP:A\rightarrow B\atop c(e_{i})=x}}}\left(f(\PP)\right)^{\Midx}\cdot\lambda(\PP)}{{\displaystyle \sum_{{\PP:A\rightarrow B\atop c(e_{i})=x}}}\lambda(\PP)}\]
 of such distributions in the trellis. We show that the complexity
of the moment calculation algorithm is $O(|\sE|)$, where $|\sE|$
is the number of edges in the trellis. 

To each edge $e\in\sE$ of the trellis $\TT$ we introduce a second
label $\fs(e)\in\sR$, which we will refer to as the \emph{c-label}.
For distinction, we will call $\lambda(e)$ the \emph{$\lambda$-label}.

\begin{example}
We continue Example \ref{exa:ChannelProbabilities}. Solid lines correspond
to the c-label $\fs(e)=1$, dashed lines correspond to $\fs(e)=-1$
(bipolar binary notation). E.g., the path $\PP=\mathtt{a}\mathtt{d}\mathtt{i}\mathtt{k}$
has the c-label $\mathbf{\fs}(\PP)=+1\,-1\,-1\,+1$ which is a code
word.
\end{example}
Let \[
g_{i}\left(\fs(e)\right):x\mapsto y;\, x,y\in\sR\]
 be a common function of $\fs(e)$ for all edges $e\in\sE_{i-1,i}$.
Further, let \[
f\left(\mathbf{\fs}(\PP)\right)=f\left(\fs(e_{1}),\fs(e_{2}),\dots,\fs(e_{L})\right):\mathbf{\fs}\mapsto y;\,\fs(e_{i}),y\in\sR\]
 be a function of the c-labels of the edges of a path $\PP$ with
length $L$. The bold letter indicates that $\mathbf{\fs}$ is a vector.
For simplicity, in the following we will abbreviate $g_{i}\left(\fs(e)\right)$
and $f\left(\mathbf{\fs}(\PP)\right)$ by $g_{i}(e)$ and $f(\PP)$,
respectively. The functions $f(\PP)$ have to fulfill the linearity
criterion \begin{equation}
f(\PP)=f(e_{1}e_{2}\cdots e_{n})=g_{1}(e_{1})+g_{2}(e_{2})+\dots+g_{n}(e_{n})\label{eq:property_of_f}\end{equation}
 for all paths $\PP:A\rightarrow B$.

\begin{definitn}[Forward Numerator]
\label{def:forward_trellis_moment}We define the $\Midx$-th \emph{forward
numerator} of a function $f$ at vertex $v$ of a trellis $T$ as\begin{equation}
\alpha^{(\Midx)}(v):=\sum_{\PP:A\rightarrow v}\left(f(\PP)\right)^{\Midx}\cdot\lambda(\PP)\label{eq:alpha_definition}\end{equation}
with initial values\[
\alpha^{(\Midx)}(A):=\left\{ \begin{array}{ccc}
1 & : & \Midx=0\\
0 & : & \Midx>0\end{array}\right.\;.\]

\end{definitn}
\medskip{}

\begin{thm}[Forward Recursion]
\label{thm:forward-moment}The $\Midx$-th forward numerator $\alpha^{(\Midx)}(v)$
of a vertex $v\in V_{i}$ on depth $i$ can be  recursively calculated
on a trellis $T$ by\begin{equation}
\alpha^{(\Midx)}(v)=\hspace{-1.5mm}\sum_{e:\mathrm{fin}(e)=v}\hspace{-1.5mm}\lambda(e)\cdot\sum_{\Mrunidx=0}^{\Midx}{\Midx \choose \Mrunidx}\left(g_{i}(e)\right)^{\Mrunidx}\cdot\alpha^{(\Midx-\Mrunidx)}\left(\mathrm{init}(e)\right)\label{eq:alpha_recursive_formula}\end{equation}
 as in Algorithm \ref{alg:Proposed-Algorithm}.
\end{thm}
\begin{algorithm}[h]
\caption{\label{alg:Proposed-Algorithm}Computation of first ($\vmmax$+1)
Forward Numerators}

\texttt{01: /{*} initialization {*}/}

\texttt{02: }$\alpha^{(0)}(A)=1;$

\texttt{03: for (m=1 to m\_max)}

\texttt{04: }~~$\alpha^{(\Midx)}(A)=0$;

\texttt{05: /{*} recursion {*}/}

\texttt{06: for (i=1 to n) \{}

\texttt{07: }~\texttt{for (}$v\in\sV_{i}$\texttt{) \{}

\texttt{08: }~~\texttt{for (m=0 to m\_max)}

\texttt{09: }~~~$\alpha^{(\Midx)}(v)={\displaystyle \sum_{e:\mathrm{fin}(e)=v}}\lambda(e)\cdot{\displaystyle \sum_{\Mrunidx=0}^{\Midx}}{\Midx \choose \Mrunidx}\left(g_{i}(e)\right)^{\Mrunidx}\cdot$

\vspace{1mm}
$\hspace{63mm}\cdot\alpha^{(\Midx-\Mrunidx)}\left(\mathrm{init}(e)\right);$

\texttt{10: }~~\texttt{\}}

\texttt{11: }~\texttt{\}}

\texttt{12: \}}
\end{algorithm}

\begin{proof}
The proof is by induction on $\mathrm{depth}(v)$. For $\mathrm{depth}(v)=1$,
it follows from the definition of a trellis that all paths from $A$
to $v$ must consist of just one edge $e$, with $\mathrm{init}(e)=A$
and $\mathrm{fin}(e)=v$. Thus the true value of $\alpha^{(\Midx)}(v)$
is the sum of the $\lambda$-labels on all edges $e$ joining $A$
to $v$, weighted by $\left(g_{1}(e)\right)^{\Midx}$.  On the other
hand, when the algorithm computes $\alpha^{(\Midx)}(v)$ on line 9,
the value it assigns to it is (because of the initialization $\alpha^{(0)}(A)=1$,
$\alpha^{(\Midx)}(v)=0$ for $\Midx>0$)\begin{eqnarray*}
\alpha^{(\Midx)}(v) & \hspace{-2mm}= & \hspace{-4mm}\sum_{e:\mathrm{fin}(e)=v}\hspace{-1mm}\lambda(e)\cdot\sum_{\Mrunidx=0}^{\Midx}{\Midx \choose \Mrunidx}\left(g_{i}(e)\right)^{\Mrunidx}\cdot\alpha^{(\Midx-\Mrunidx)}\left(\mathrm{init}(e)\right)\\
 & \hspace{-2mm}= & \hspace{-2mm}\sum_{e:A\rightarrow v}\lambda(e)\cdot\left(g_{1}(e)\right)^{\Midx}\cdot1\end{eqnarray*}
which is, as required, the sum of the labels on all edges $e$ joining
$A$ to $v$, weighted by $\left(g_{1}(e)\right)^{\Midx}$. Thus the
algorithm works correctly for all vertices $v$ with $\mathrm{depth}(v)=1$
and any $\Midx\geq0$.

Assuming now that the assertion is true for all vertices at depth
$i$ or less and all $\Midx\leq M$, a vertex $v$ at depth $i+1$
is considered. When the algorithm computes $\alpha^{(\Midx)}(v)$
on line $9$, the value it assigns to it is \begin{equation}
\alpha^{(\Midx)}(v)=\sum_{e:\mathrm{fin}(e)=v}\lambda(e)\cdot\sum_{\Mrunidx=0}^{\Midx}{\Midx \choose \Mrunidx}\left(g_{i}(e)\right)^{\Mrunidx}\cdot\alpha^{(\Midx-\Mrunidx)}\left(\mathrm{init}(e)\right)\;.\label{eq:proof_line_star}\end{equation}
But $\mathrm{depth}(\mathrm{init}(e))=i$ and so by the induction
hypothesis\begin{equation}
\alpha^{(\Midx)}(\mathrm{init}(e))=\sum_{\PP:A\rightarrow\mathrm{init}(e)}\lambda(\PP)\cdot\left(f(\PP)\right)^{\Midx}\;.\label{eq:proof_induction_hypothesis}\end{equation}
Combining (\ref{eq:proof_line_star}) and (\ref{eq:proof_induction_hypothesis}),
we have{\small \begin{eqnarray*}
\alpha^{(\Midx)}(v) & \hspace{-2mm}= & \hspace{-4.5mm}\sum_{e:\mathrm{fin}(e)=v}\hspace{-2mm}\lambda(e)\cdot\sum_{\Mrunidx=0}^{\Midx}{\Midx \choose \Mrunidx}\left(g_{i}(e)\right)^{\Mrunidx}\cdot\hspace{-4mm}\sum_{\PP:A\rightarrow\mathrm{init}(e)}\hspace{-3.5mm}\lambda(\PP)\cdot\left(f(\PP)\right)^{\Midx-\Mrunidx}\\
 & \hspace{-2mm}= & \hspace{-4.5mm}\sum_{e:\mathrm{fin}(e)=v}\sum_{\PP:A\rightarrow\mathrm{init}(e)}\hspace{-4mm}\lambda(e)\hspace{-0.5mm}\cdot\hspace{-0.5mm}\lambda(\PP)\hspace{-0.5mm}\cdot\hspace{-0.5mm}\sum_{\Mrunidx=0}^{\Midx}{\Midx \choose \Mrunidx}\hspace{-0.5mm}\left(g_{i}(e)\right)^{\Mrunidx}\hspace{-1mm}\cdot\hspace{-0.5mm}\left(f(\PP)\right)^{\Midx-\Mrunidx}\hspace{-1.5mm}.\end{eqnarray*}
}Using the binomial theorem we obtain\begin{equation}
\alpha^{(\Midx)}(v)=\hspace{-2mm}\sum_{e:\mathrm{fin}(e)=v}\sum_{\PP:A\rightarrow\mathrm{init}(e)}\lambda(\PP e)\cdot\left(f(\PP)+g_{i}(e)\right)^{\Midx}.\label{eq:proof_last_line}\end{equation}
But every path from $A$ to $v$ must be of the form $\PP e$, where
$\PP$ is a path from $A$ to a vertex $u$ with $\mathrm{depth}(u)=i$,
$\mathrm{init}(e)=u$ and $\mathrm{fin}(e)=v$. Thus by (\ref{eq:proof_last_line}),
$\alpha^{(\Midx)}(v)$ is correctly calculated by the algorithm.
\end{proof}
\medskip{}

\begin{remrk}[Flow]
$\alpha^{(0)}(v)$ in (\ref{eq:alpha_definition}) is the flow $\vFlow{A}{v}$
from $A$ to $v$ (cf. Definition \ref{def:flow}) as it is calculated
by the BCJR algorithm.

\end{remrk}

\begin{remrk}
$f$ and $g_{\Mdeptha}$ do not necessarily have to be scalars. Theorem
\ref{thm:forward-moment} holds for all separable linear functions
$f$ fulfilling Equation (\ref{eq:property_of_f}).
\end{remrk}
\begin{thm}[Complexity]
\label{thm:MomentCalculationComplexity}The proposed moment computing
algorithm requires $O(|\sE|)$ arithmetic operations, i.e. multiplications
and additions.
\end{thm}
\begin{proof}
The calculation of the powers of $g_{i}(e)$ up to a maximum moment
$\vmmax$ for all edges $e\in\sE$ requires $|\sE|\cdot\max(\vmmax-1,0)$
multiplications and no additions. We do not consider the operations
needed for calculating $g_{i}(e)$ here. The execution of the sum
term over $l$ in line $9$ of the algorithm requires $\Midx$ additions,
$2\Midx+1$ multiplications for $\Midx>0$ %
\footnote{For $\Mrunidx=0$, $\left(g_{i}(e)\right)^{\Mrunidx}=1$ and thus
only one multiplication is necessary.%
} and no multiplications for $\Midx=0$. Therefore line $9$ requires
\[
\rho^{-}(v)\cdot[1+2\Midx+1]=\rho^{-}(v)\cdot2(\Midx+1)\]
 multiplications for $\Midx>0$, $\rho^{-}(v)$ multiplications for
$\Midx=0$, and $\rho^{-}(v)-1+\rho^{-}(v)\cdot\Midx$ additions.
Hence, for a vertex $v\in V_{i}$, \[
\rho^{-}(v)+\sum_{\Midx=1}^{\vmmax}\rho^{-}(v)\cdot2(\Midx+1)=\rho^{-}(v)\cdot(\vmmax^{2}+3\vmmax+1)\]
  multiplications and {\small \[
\sum_{\Midx=0}^{\vmmax}\left(\rho^{-}(v)\cdot(\Midx+1)-1\right)=\rho^{-}(v)\cdot\left(\frac{1}{2}\vmmax^{2}+\frac{3}{2}\vmmax+1\right)-(\vmmax+1)\]
}additions are necessary. The total number of multiplications required
by the algorithm is thus\begin{equation}
\mathtt{mult}=(\vmmax^{2}+3\vmmax+1)\cdot\sum_{i=1}^{n}\sum_{v\in V_{i}}\rho^{-}(v)\label{eq:multiplications_sum}\end{equation}
and the total number of additions is\begin{eqnarray}
\mathtt{add} & \hspace{-1mm}= & \hspace{-1mm}\sum_{i=1}^{n}\sum_{v\in\sV_{i}}\left(\rho^{-}(v)\cdot\left(\frac{1}{2}\vmmax^{2}+\frac{3}{2}\vmmax+1\right)-(\vmmax-1)\right)\nonumber \\
 & \hspace{-1mm}= & \hspace{-1mm}\left(\frac{1}{2}\vmmax^{2}+\frac{3}{2}\vmmax+1\right)\cdot\sum_{i=1}^{n}\sum_{v\in\sV_{i}}\rho^{-}(v)-\nonumber \\
 &  & \hspace{30mm}-(\vmmax-1)\cdot\sum_{i=1}^{n}\sum_{v\in\sV_{i}}1\;.\label{eq:additions_sum}\end{eqnarray}
Every edge in $\sE$ is counted exactly once in the sum in (\ref{eq:multiplications_sum}),
since if $e:u\rightarrow v$, then $\mathrm{fin}(e)\in\sV_{i}$ for
exactly one value of $i\in\{1,2,\dots,\Trank\}$. Thus the sum in
(\ref{eq:multiplications_sum}) is $|\sE|$. The second sum in (\ref{eq:additions_sum})
is $|\sV|-1$, since every vertex except $A$ is in $\bigcup_{i=1}^{\Trank}\sV_{i}$.
Thus from (\ref{eq:multiplications_sum}) and (\ref{eq:additions_sum}),
we have\begin{eqnarray*}
\mathtt{mult} & \hspace{-1mm}= & \hspace{-1mm}(\vmmax^{2}+3\vmmax+1)\cdot|\sE|\\
\mathtt{add} & \hspace{-1mm}= & \hspace{-1mm}\left(\frac{1}{2}\vmmax^{2}+\frac{3}{2}\vmmax+1\right)\cdot|\sE|-(\vmmax+1)\cdot(|\sV|-1)\end{eqnarray*}
so that the total number of arithmetic operations required by the
algorithm is \begin{eqnarray*}
\left(\frac{3}{2}\vmmax^{2}+\frac{9}{2}\vmmax+2\right)\cdot|\sE|-(\vmmax+1)\cdot|\sV|+\vmmax+1\\
\leq\left(\frac{3}{2}\vmmax^{2}+\frac{9}{2}\vmmax+2\right)\cdot|\sE|.\end{eqnarray*}
 We have $|\sV|\geq1$, and $|\sE|-|\sV|+1\geq0$ (since the trellis
is connected), so that the total number of operations required is
bounded above by $\left(\frac{3}{2}\vmmax^{2}+\frac{9}{2}\vmmax+2\right)\cdot|\sE|$
and bounded below by $\left(\frac{3}{2}\vmmax^{2}+\frac{7}{2}\vmmax+1\right)\cdot|\sE|$
(disregarding the complexity of the computation of $g_{i}(e)$).
\end{proof}
\medskip{}

In analogy to the forward numerator in Definition \ref{def:forward_trellis_moment}
we can also define a backward numerator.

\begin{definitn}[Backward Numerator]
The $\Midx$-th \emph{backward numerator} of a vertex $v\in\sV_{i}$
is defined as \[
\beta^{(\Midx)}(v):=\sum_{\PP:v\rightarrow B}\left(f(\PP)\right)^{\Midx}\cdot\lambda(\PP)\]
with initial values\[
\beta^{(\Midx)}(B)=\left\{ \begin{array}{ccc}
1 & : & \Midx=0\\
0 & : & \Midx>0\end{array}\right.\;.\]

\end{definitn}
\begin{thm}[Backward Recursion]
\label{thm:BackwardMoment}The $\Midx$-th backward numerator $\beta^{(\Midx)}(v)$
of a vertex $v\in\sV_{i}$ can be calculated in a trellis $T$ by\[
\beta^{(\Midx)}(v)=\hspace{-1.5mm}\sum_{e:\mathrm{init}(e)=v}\hspace{-2mm}\lambda(e)\cdot\sum_{\Mrunidx=0}^{\Midx}{\Midx \choose \Mrunidx}\left(g_{i+1}(e)\right)^{\Mrunidx}\cdot\beta^{(\Midx-\Mrunidx)}\left(\mathrm{fin}(e)\right).\]

\end{thm}
\begin{proof}
The proof is analog to the proof of Theorem \ref{thm:forward-moment}.
\end{proof}
\medskip{}

It obviously holds that $\alpha^{(\Midx)}(B)=\beta^{(\Midx)}(A)=:\Mtrellis{\Midx}{\TT}$,
providing the \emph{m}-th \emph{moment} \[
\MNtrellis{\Midx}{\TT}:=\frac{\Mtrellis{\Midx}{\TT}}{\Mtrellis{0}{\TT}}=\frac{{\displaystyle \sum_{\PP:A\rightarrow B}}\left(f(\PP)\right)^{\Midx}\cdot\lambda(\PP)}{{\displaystyle \sum_{\PP:A\rightarrow B}}\lambda(\PP)}\]
 of the distribution of function $f$ given $\TT$.

In analogy to the BCJR algorithm \cite{BCJR1974} for calculating
symbol probabilities, we next consider the calculation of moments
of $f$ introducing a constraint on the value of the c-labels at a
certain depth $i$ in the trellis. I.e., the moments are calculated
in a sub-trellis of $\TT$.

\begin{definitn}[Symbol Moment]
\label{def:SymbolMoment}We define the $\Midx$-th \emph{symbol
moment} $\MNsymbol{\Midx}{i}{\TT,\MsymS}$ at depth $i$ of a trellis
$\TT$ as \[
\MNsymbol{\Midx}{i}{\TT,\MsymS}:=\frac{{\displaystyle \sum_{{\PP:A\rightarrow B\atop c_{i}=x}}}\left(f(\PP)\right)^{\Midx}\cdot\lambda(\PP)}{{\displaystyle \sum_{{\PP:A\rightarrow B\atop c_{i}=x}}}\lambda(\PP)}\]
 where $c_{i}=c(e_{i})$and $e_{i}\in\sE_{i-1,i}$ is the i-th edge
of path $\PP$. 
\end{definitn}

\begin{thm}
\label{thm:symbol-moment}The $\Midx$-th symbol moment can be calculated
by\[
\MNsymbol{\Midx}{i}{\TT,\MsymS}=\frac{\Msymbol{\Midx}{i}{\TT,\MsymS}}{\Msymbol{0}{i}{\TT,\MsymS}}\]
 with \begin{multline}
\Msymbol{\Midx}{i}{\TT,\MsymS}=\sum_{{e\in\sE_{i-1,i}:\atop \fs(e)=\MsymS}}\lambda(e)\cdot\sum_{\Mrunidx=0}^{\Midx}{\Midx \choose \Mrunidx}\beta^{(\Midx-\Mrunidx)}(\mathrm{fin}(e))\cdot\\
\cdot\sum_{k=0}^{\Mrunidx}{\Mrunidx \choose k}\left(g_{i}(e)\right)^{k}\cdot\alpha^{(\Mrunidx-k)}(\mathrm{init}(e)).\label{eq:Theorem-SymbolMoment}\end{multline}

\end{thm}
\medskip{}

\begin{proof}
Let $\PP_{H}$ and $\PP_{T}$ denote the head and tail parts of the
paths $\PP:A\rightarrow B$ through the trellis $\TT$, with an edge
$e$ in between, i.e., $\PP=\PP_{H}e\PP_{T}$ with $\mathrm{init}(\PP_{H})=A$,
$\mathrm{fin}(\PP_{H})=\mathrm{init}(e)$, $\mathrm{fin}(e)=\mathrm{init}(\PP_{T})$
and $\mathrm{fin}(\PP_{T})=B$, for a given depth $i$ and $e\in\sE_{i-1,i}$.
Then we can write\begin{multline*}
\Msymbol{\Midx}{i}{\TT,\MsymS}={\displaystyle \sum_{{\PP:A\rightarrow B\atop c_{i}=x}}}\left(f(\PP)\right)^{\Midx}\cdot\lambda(\PP)\\
=\hspace{-1mm}\sum_{{e\in\sE_{i-1,i}:\atop \fs(e)=\MsymS}}\sum_{{\PP_{H}:A\rightarrow\atop \mathrm{init}(e)}}\sum_{{\PP_{T}:\mathrm{fin}(e)\atop \rightarrow B}}\hspace{-0mm}\left(f(\PP_{H})\hspace{-.5mm}+\hspace{-.5mm}g_{i}(e)\hspace{-.5mm}+\hspace{-.5mm}f(\PP_{T})\right)^{\Midx}\cdot\\
\cdot\lambda\left(\PP_{H}\, e\,\PP_{T}\right).\end{multline*}
Applying \noun{Bayes}' rule twice and separating the $\lambda$-labels
we obtain \begin{multline*}
\Msymbol{\Midx}{i}{\TT,\MsymS}=\sum_{{e\in\sE_{i-1,i}:\atop \fs(e)=\MsymS}}\lambda(e)\cdot\sum_{\Mrunidx=0}^{\Midx}{\Midx \choose \Mrunidx}\left(f(\PP_{T})\right)^{\Midx-\Mrunidx}\cdot\lambda(\PP_{T})\\
\cdot\sum_{k=0}^{\Mrunidx}{\Mrunidx \choose k}\left(g_{i}(e)\right)^{k}\cdot\left(f(\PP_{H})\right)^{\Mrunidx-k}\cdot\lambda(\PP_{H})\end{multline*}
 and using the definitions of forward and backward numerators finally
yields  the assertion of the theorem.
\end{proof}
\medskip{}

\begin{thm}[Computational Complexity]
\label{thm:SymbolMomentComplexity}Given the forward numerators $\alpha^{(\Midx)}(v)$
and the backward numerators $\beta^{(\Midx)}(v)$ up to order $\Midx$
for all $v\in\sV$, the computation of $\Msymbol{\Midx}{i}{\TT,\MsymS}$
for all $i\in1\dots\Trank$  requires $O(|\sE|)$ arithmetic operations. 
\end{thm}
\begin{proof}
Consider Equation (\ref{eq:Theorem-SymbolMoment}). The sum over $k$
requires $2\Mrunidx$ multiplications and $\Mrunidx$ additions. The
sum over $\Mrunidx$ requires \[
\sum_{\Mrunidx=0}^{\Midx}(2\Mrunidx+2)-1=\Midx^{2}+3\Midx+1\]
 multiplications and \[
\sum_{\Mrunidx=0}^{\Midx}\Mrunidx+\left(\Midx-1\right)=\frac{1}{2}\left(\Midx^{2}+3\Midx-2\right)\]
 additions. There are at most $|\sE_{i-1,i}|$ edges $e$ for which
$e\in\sE_{i-1,i}$ and $\fs(e)=\MsymS$, thus the sum over these edges
requires at most $|\sE_{i-1,i}|\cdot\left((\Midx^{2}+3\Midx+1)+1\right)$
multiplications and $|\sE_{i-1,i}|\cdot\frac{1}{2}\left(\Midx^{2}+3\Midx-2\right)+|\sE_{i-1,i}|-1$
additions. As we calculate the symbol moments for all $i\in1\dots\Trank$,
we can finally upper limit the requirements by \begin{eqnarray*}
\mathtt{mult} & \leq & |\sE|\cdot(\Midx^{2}+3\Midx+2)\\
\mathtt{add} & \leq & |\sE|\cdot0.5\left(\Midx^{2}+3\Midx\right)-\Trank.\end{eqnarray*}

\end{proof}
\begin{remrk}[Forward/Backward Moments]
For numeric reasons it may be advantageous to directly compute the
\emph{forward} and \emph{backward moments} \[
\bar{\alpha}^{(m)}(v):=\frac{\alpha^{(m)}(v)}{\alpha^{(0)}(v)}\quad\mathrm{and}\quad\bar{\beta}^{(m)}(v):=\frac{\beta^{(m)}(v)}{\beta^{(0)}(v)},\]
 respectively, and to calculate and carry the 0-th numerators (flows)
in the logarithmic domain.

\medskip{}

\end{remrk}
Finally in this section, we describe the calculation of distributions
over all paths $\PP:A\rightarrow B$, or a subset of paths, in the
trellis in analogy to the calculation of moments and symbol moments,
respectively.

\begin{definitn}[Forward/Backward Distribution]
\label{def:ForwardBackwardDistribution}We define the \emph{forward
distribution} $\alpha^{\fD}(v)$ and the \emph{backward distribution}
$\beta^{\fD}(v)$ at a vertex $v$ as the mapping functions \[
\vu\mapsto\sum_{{\PP:A\rightarrow v\atop f(\PP)=\vu}}\lambda(\PP)\qquad\mathrm{and}\qquad\vu\mapsto\sum_{{\PP:v\rightarrow B\atop f(\PP)=\vu}}\lambda(\PP),\]
 respectively.
\end{definitn}
\begin{thm}
\label{thm:Distributions}The forward distribution $\alpha^{\fD}(v)$
at a vertex $v\in\sV_{i}$ can be recursively calculated in the trellis
by\[
\alpha^{\fD}(v)=\sum_{e:\mathrm{fin}(e)=v}\left(\alpha^{\fD}\left(\mathrm{init}(e)\right)\boxplus g_{i}(e)\right)\cdot\lambda(e),\]
where $a(u)\boxplus b$ denotes a shift of the domain of the distribution
$a(u)$ by $b$, and $\alpha^{\fD}(A)$ equals the Dirac function.
The calculation of $\beta^{\fD}(v)$ is analog with $\beta^{\fD}(B)$
being the Dirac function. The distribution $\Dtrellis{\TT}$ and the
\emph{symbol distribution} $\Dsymbol{i}{\MsymS,\TT}$ can be calculated
by\[
\theta^{\fD}(\TT)=\sum_{v\in V_{i}}\alpha^{\fD}\left(v\right)*\beta^{\fD}\left(v\right)\]
and \[
\Omega_{i}^{\fD}(\TT,\MsymS)=\hspace{-3mm}\sum_{{e\in\sE_{i-1,i}:\atop \fs(e)=\MsymS}}\hspace{-3mm}\left(\alpha^{\fD}\left(\mathrm{init}(e)\right)\boxplus g_{i}(e)\right)*\beta^{\fD}\left(\mathrm{fin}(e)\right)\cdot\lambda(e),\]
respectively. Herby, $*$ denotes the convolution operator, i.e. for
two distributions $a(u)$ and $b(u)$ it holds\[
a(u)*b(u)=\int_{-\infty}^{\infty}a(\nu)\cdot b(u-\nu)\, d\nu\,.\]

\end{thm}
\begin{proof}
Theorem \ref{thm:Distributions} follows directly from Definition
\ref{def:ForwardBackwardDistribution}. 
\end{proof}
\medskip{}

\begin{remrk}[Density Distributions]
When normalizing distributions by the corresponding flow, we obtain
\emph{density distributions}.
\end{remrk}

\begin{remrk}[Probability Density Functions]
For $\lambda(e)$ being probabilities, normalized distributions are
probability density functions with the mapping $f(\PP)\rightarrow P\left(f(\PP)\right)$
 and $\sum_{f(\PP)}P\left(f(\PP)\right)=1$. 
\end{remrk}

\begin{remrk}
By Theorem \ref{thm:Distributions}, the complexity due to the calculation
on the trellis is in general not reduced (except for the hard decision
case) as infinite resolution of the domain of $\alpha^{\fD}(v)$ etc.
is required.. However, in Appendix \ref{sub:Calculate-Actual-Distribution}
an algorithm is introduced which approximates Theorem \ref{thm:Distributions}
and does reduce complexity.
\end{remrk}

\begin{remrk}
We cannot only determine the distribution and its moments of a trellis
or sub-trellis, but also of a single edge.
\end{remrk}

\begin{remrk}
The symbol distribution for two sub-trellises of the $[7\;5]_{\mathrm{oct}}$
convolutional code, namely the sub-codes with the i-th code bit $c_{i}=+1$
and $c_{i}=-1$, respectively, is given in Example \ref{exa:SymbolDistribution}.
The curves obtained by Gaussian approximation almost coincide with
the ones plotted in Figure \ref{fig:Symbol-distribution}.
\end{remrk}

\begin{remrk}
It is straight forward to extend the proposed algorithm to the calculation
of joint moments of two or more functions. E.g., \[
\bar{\theta}_{y,z}^{(k,m)}:=\frac{\sum_{\PP}\left(f_{y}(\PP)\right)^{k}\left(f_{z}(\PP)\right)^{m}\cdot\lambda(\PP)}{\sum_{\PP}\lambda(\PP)}\]
 can be calculated using  \begin{eqnarray*}
\alpha_{y,z}^{(k,m)}(v) & := & \sum_{\PP:A\rightarrow v}\left(f_{y}(\PP)\right)^{k}\cdot\left(f_{z}(\PP)\right)^{m}\cdot\lambda(\PP)\\
 & = & {\displaystyle \sum_{e:\mathrm{fin}(e)=v}}\lambda(e)\cdot\sum_{j=0}^{k}\sum_{l=0}^{m}{k \choose j}{m \choose l}\\
 &  & \qquad g_{i,y}^{k-j}(e)\cdot g_{i,z}^{m-l}(e)\cdot\alpha_{y,z}^{(j,l)}\left(\mathrm{init}(e)\right)\end{eqnarray*}
 with $i=\mathrm{depth}(v)$.
\end{remrk}

\section{\label{sec:Applications}Applications}

We will now apply the results of Section \ref{sec:Trellis-Computations}
to linear block codes. We compute the moments \[
\fE_{\mathbf{C}}\left[\left(H(\mathbf{c}|\mathbf{w})\right)^{\Midx}|\mathbf{r},c_{i}=x\right]:=\sum_{\mathbf{c}\in\sC}\left(H(\mathbf{c}|\mathbf{w})\right)^{\Midx}P(\mathbf{c}|\mathbf{r},c_{i}=x)\]
 of the distribution  \[
\mathcal{D}:\vu=H(\mathbf{c}|\mathbf{w})\mapsto P\left(\vu|\mathbf{r},c_{i}=x\right)=\hspace{-2mm}\sum_{{\mathbf{c}\in\sC:\atop H(\mathbf{c}|\mathbf{w})=\vu}}\hspace{-2mm}P\left(\mathbf{c}|\mathbf{r},c_{i}=x\right)\]
over all code words $\mathbf{c}\in\sC$ given a received word $\mathbf{r}$
and the \emph{i}-th code bit being $c_{i}=x\in\{-1,1\}$, where \[
H(\mathbf{c}|\mathbf{w})=-\log P(\mathbf{c}|\mathbf{w})\]
 is the conditional uncertainty of $\mathbf{c}$ given a word $\mathbf{w}$
and $P(\mathbf{c}|\mathbf{r})$ is the conditional probability of
$\mathbf{c}$ given $\mathbf{r}$. These moments are required, e.g.,
for the discriminated belief propagation algorithm in \cite{Sorger2007}.
As a special case we can calculate the conditional mean uncertainty
or \emph{entropy} \[
H(\mathbf{C}|\mathbf{r})=\sum_{\mathbf{c}\in\sC}H(\mathbf{c}|\mathbf{r})\cdot P(\mathbf{c}|\mathbf{r})\]
 of a code or sub-code given $\mathbf{r}$.

Both for hard decision (BSC) and soft decision (AWGN channel) the
conditional uncertainty is linearly related to the correlation $\mathbf{cw^{T}}$
(cf. Appendix \ref{sub:Relation-uncertainty-correlation}),\begin{equation}
H(\mathbf{c}|\mathbf{w})=\KI+\KII\cdot\mathbf{c}\mathbf{w}^{T}\;,\label{eq:Relation-Uncertainty-Correlation}\end{equation}
with $\KI$ and $\KII$ being constant functions of error probability
and vector $\mathbf{w}$ (assuming equiprobable code words). Therefore,
when applying the binomial theorem,  \begin{eqnarray*}
{\displaystyle {\displaystyle \sum_{\mathbf{c}\in\sC}}}\left(H(\mathbf{c}|\mathbf{w})\right)^{\Midx}\cdot P(\mathbf{c}|\mathbf{r},c_{i}=x)\\
 & \hspace{-80mm}= & \hspace{-40mm}{\displaystyle {\displaystyle \sum_{\mathbf{c}\in\sC}}}\left(\KI+\KII\cdot\mathbf{c}\mathbf{w}^{T}\right)^{\Midx}\cdot P(\mathbf{c}|\mathbf{r},c_{i}=x)\\
 & \hspace{-80mm}= & \hspace{-40mm}{\displaystyle {\displaystyle \sum_{\mathbf{c}\in\sC}}}\sum_{\Mrunidx=0}^{\Midx}{\Midx \choose \Mrunidx}\KI^{\Midx-\Mrunidx}\KII^{\Mrunidx}\left(\mathbf{c}\mathbf{w}^{T}\right)^{\Mrunidx}P(\mathbf{c}|\mathbf{r},c_{i}=x)\\
 & \hspace{-80mm}= & \hspace{-40mm}\sum_{\Mrunidx=0}^{\Midx}{\Midx \choose \Mrunidx}\KI^{\Midx-\Mrunidx}\KII^{\Mrunidx}\cdot\fE_{\mathbf{C}}\left[\left(\mathbf{c}\mathbf{w}^{T}\right)^{\Mrunidx}|\mathbf{r},c_{i}=x\right],\end{eqnarray*}
 it is sufficient to calculate the moments \begin{equation}
\fE_{\mathbf{C}}\left[\left(\mathbf{c}\mathbf{w}^{T}\right)^{\Midx}|\mathbf{r},c_{i}=x\right]={\displaystyle {\displaystyle \sum_{\mathbf{c}\in\sC}}}\left(\mathbf{c}\mathbf{w}^{T}\right)^{\Midx}\cdot P(\mathbf{c}|\mathbf{r},c_{i}=x)\label{eq:Correlation-Moments}\end{equation}
 of the correlation $\mathbf{cw}^{T}$ on the trellis which will be
done in the following.

Consider a binary linear block code $\sC$ of length $\Trank$ which
is representable in a trellis, e.g., a terminated convolutional code.
 Let the c-labels $\fs(e)=c_{i}\in\{\pm1\}$ be the bipolar representation
of the code bit labeling edge $e$. To each path $\PP:A\rightarrow B$
it belongs a sequence $\mathbf{c}(\PP)$ of $\Trank$ c-labels representing
a code word $\mathbf{c}\in\sC$. Let $\mathbf{r}=[r_{1}r_{2}\cdots r_{\Trank}]$,
$r_{i}\in\mathbb{R}$, be the noisy version of a code word $\mathbf{c}$
after transmission over a memory-less channel. Let the $\lambda$-label
of a path $\PP$ be the conditional probability of the received word
$\mathbf{r}$ given the code word $\mathbf{c}$, i.e., $\lambda(\PP)=P(\mathbf{r}|\mathbf{c})$.
Let further the function $f$ of the paths' c-labels, i.e., the function
of the code words, be the correlation (inner product) of $\mathbf{w}$
and $\mathbf{c}$, \[
f(\PP)=f(\mathbf{c}(\PP))=\mathbf{c}\mathbf{w}^{T}=\sum_{i=1}^{n}c_{i}w_{i}\,.\]
Hence, $g_{i}(e)=c_{i}w_{i}$ and the separability criterion (\ref{eq:property_of_f})
is fulfilled. In the trellis of $\sC$, for each vertex $v\in\sV$
the c-labels $\fs(e)$ of edges $\{e:\mathrm{init}(e)=v\}$ emerging
from $v$ are distinct. Therefore there is a one-to-one mapping of
each code word $\mathbf{c}$ to a path $\PP$ in the trellis, and
we can apply the theorems of Section \ref{sec:Trellis-Computations}
replacing $\sum_{\PP}$ by $\sum_{\mathbf{c}}$. Applying \noun{Bayes}'
rule to (\ref{eq:Correlation-Moments}), \begin{equation}
\fE_{\mathbf{C}}\left[\left(\mathbf{c}\mathbf{w}^{T}\right)^{\Midx}|\mathbf{r},c_{i}=x\right]=\frac{{\displaystyle {\displaystyle \sum_{\mathbf{c}\in\sC:c_{i}=x}}}\left(\mathbf{c}\mathbf{w}^{T}\right)^{\Midx}P(\mathbf{r}|\mathbf{c})}{{\displaystyle {\displaystyle \sum_{\mathbf{c}\in\sC:c_{i}=x}}}P(\mathbf{r}|\mathbf{c})},\label{eq:Correlation-Moments-02}\end{equation}
 and comparing with Definition \ref{def:SymbolMoment} we observe
that Theorems \ref{thm:forward-moment} and \ref{thm:BackwardMoment}
hold, and hence these moments can be calculated in the trellis according
to Theorem \ref{thm:symbol-moment} as the symbol moments \[
\fE_{\mathbf{C}}\left[\left(\mathbf{c}\mathbf{w}^{T}\right)^{\Midx}|\mathbf{r},c_{i}=x\right]=\MNsymbol{\Midx}{i}{\MsymS}.\]

Analogously, when omitting the code bit constraint $c_{i}=x$, the
moments are given by \[
\fE_{\mathbf{C}}\left[\left(\mathbf{c}\mathbf{w}^{T}\right)^{\Midx}|\mathbf{r}\right]=\sum_{\mathbf{c}\in\sC}\left(\mathbf{c}\mathbf{w}^{T}\right)^{\Midx}\cdot P(\mathbf{c}|\mathbf{r})=\MNtrellis{\Midx}{\TT}.\]

For $\mathbf{w}=\mathbf{r}$, $m=1$ and $g_{i}(e)=c_{i}r_{i}$ we
can thus calculate the conditional entropies \[
H(\sC|\mathbf{r})=\sum_{\mathbf{c}\in\sC}H(\mathbf{c}|\mathbf{r})\cdot P(\mathbf{c}|\mathbf{r})=\KI+\KII\cdot\MNtrellis{1}{\TT}\]
 and \[
H(\sC_{i}(x)|\mathbf{r})=\hspace{-2mm}\sum_{\mathbf{c}\in\sC:c_{i}=x}\hspace{-3mm}H(\mathbf{c}|\mathbf{r})\cdot P(\mathbf{c}|\mathbf{r})=\KI+\KII\cdot\MNsymbol{1}{i}{\MsymS}\]
 of the code $\sC$ and the sub-code $\sC_{i}(x)=\{\mathbf{c}\in\sC:c_{i}=x\}$
given $\mathbf{r}$, respectively. While $H(\sC|\mathbf{r})$ can
also be calculated with the classical BCJR algorithm as \begin{eqnarray*}
\sum_{\mathbf{c}\in\sC}\mathbf{c}\mathbf{r}^{T}\cdot P(\mathbf{c}|\mathbf{r}) & \hspace{-2mm}= & \hspace{-2mm}\sum_{i=1}^{\Trank}\sum_{\mathbf{c}\in\sC}c_{i}r_{i}\cdot P(\mathbf{c}|\mathbf{r})\\
 & \hspace{-2mm}= & \hspace{-2mm}\sum_{i=1}^{\Trank}r_{i}\cdot\left(\sum_{{\mathbf{c}\in\sC:\atop c_{i}=1}}P(\mathbf{c}|\mathbf{r})-\hspace{-2mm}\sum_{{\mathbf{c}\in\sC:\atop c_{i}=-1}}P(\mathbf{c}|\mathbf{r})\right),\end{eqnarray*}
this does not hold for the conditional entropy of $\sC_{i}(x)$.

\begin{remrk}
For a convolutional code with $c$ outputs, to each edge in the trellis
are assigned $c$ code symbols. To apply our definition of a single
symbol label per edge, each edge $e$ of the original trellis is replaced
by a path $e'_{1}e'_{2}\cdots e'_{c}$ of $c$ edges which fulfill
\[
\mathrm{init}(e)=\mathrm{init}(e_{1}'),\mathrm{fin}(e_{1}')=\mathrm{init}(e_{2}'),\dots,\mathrm{fin}(e_{c}')=\mathrm{fin}(e)\]
 and to each edge $e_{i}'$ one code symbol is assigned.
\end{remrk}
\begin{example}
Figure \ref{fig:Symbol-distribution} shows the distribution of $P\left(\mathbf{c}\mathbf{r}^{T},c_{i}=\pm1|\mathbf{r}\right)$
over $\mathbf{cr}^{T}$ for the $[5\,7]_{oct}$ convolutional code
of length $\Trank=200$ given a noisy received word $\mathbf{r}$
after transmission over a BSC with bit error probability $p=0.35$.
These are the normalized symbol distributions $\DsymbolN{i=10}{\pm1}$
weighted by the probability $P(c_{i}=\pm1|\mathbf{r})$.
\end{example}

\begin{example}
Figure \ref{fig:Trellis-distribution}%
\begin{figure}
\begin{centering}
\includegraphics{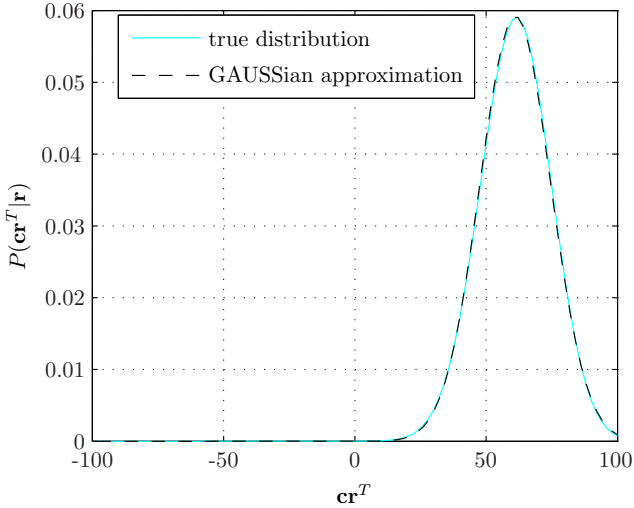}
\par\end{centering}

\caption{\label{fig:Trellis-distribution}Distribution and its \noun{Gauss}ian
Approximation}

\end{figure}
 shows a distribution of the terminated $[7\;5]_{\mathrm{oct}}$ convolutional
code  as well the \noun{Gauss}ian approximation given the first two
moments for a BSC.
\end{example}

\section{Conclusions}

A trellis represents a general distribution which can be marginalized,
e.g. with respect to edge labels. Two algorithms for computations
on the trellis were presented: One allowing to calculate distributions,
the other to compute their moments, allowing to approximate the distributions.
The latter was derived by generalizing the forward/backward recursion
as known from the BCJR algorithm. The results were transferred to
the concrete problem of computing the moments of the conditional distribution
of the correlation between a block code and some given word. The moment
calculation algorithm is a requirement for efficient implementation
of the discriminated belief propagation algorithm in \cite{Sorger2007}.
It can also be used to calculate the conditional entropy of a code
or sub-code. Though not the focus of this paper, in the Appendix it
is shown that the algorithm does not restrict to calculation with
real numbers, but is valid for any commutative semi-ring, thus providing
a generalization of the Viterbi algorithm. The asymptotic complexity
of the moment computation algorithm is the same as for the BCJR algorithm.

\appendix

\subsection{\label{sub:Relation-uncertainty-correlation}Relation between Uncertainty
and Correlation}

The conditional uncertainty of a code word $\mathbf{c}$ given a word
$\mathbf{w}$ is defined as \[
H(\mathbf{c}|\mathbf{w}):=-\log_{2}P(\mathbf{c}|\mathbf{w})=-\log_{2}P(\mathbf{w}|\mathbf{c})+\underbrace{\log_{2}\frac{P(\mathbf{w})}{P(\mathbf{c})}}_{\KIa}\]
where $\KIa$ is a constant assuming equiprobable code words. Assuming
further that $w_{\Mdeptha}$ is independent of $c_{\Mdepthb}$ for
$\Mdeptha\neq\Mdepthb$ it follows that\begin{eqnarray*}
\log_{2}P(\mathbf{w}|\mathbf{c}) & = & \log_{2}\prod_{\Mdeptha=1}^{\Trank}P(w_{\Mdeptha}|c_{\Mdeptha})\\
 & = & \sum_{\Mdeptha=1}^{\Trank}\log_{2}P(w_{\Mdeptha}|c_{\Mdeptha})\;.\end{eqnarray*}

\begin{itemize}
\item For a binary symmetric channel (BSC) with $w_{\Mdeptha},c_{\Mdeptha}\in\{\pm1\}$
and error probability $p$ the Hamming distance between $\mathbf{c}$
and $\mathbf{w}$ is $\frac{n-\mathbf{cw}^{T}}{2}$ which gives \begin{eqnarray*}
\log_{2}P(\mathbf{w}|\mathbf{c}) & \hspace{-2.5mm}= & \hspace{-2.5mm}\frac{n-\mathbf{cw}^{T}}{2}\log_{2}p+\frac{n+\mathbf{cw}^{T}}{2}\log_{2}(1-p)\\
 & \hspace{-2.5mm}= & \hspace{-2.5mm}\underbrace{\frac{n}{2}\log_{2}\left(p(1-p)\right)}_{\KIb}+\mathbf{cw}^{T}\cdot\underbrace{\frac{1}{2}\log_{2}\frac{1-p}{p}}_{\KII}.\end{eqnarray*}

\item For an AWGN channel with noise variance $\sigma^{2}$ we obtain (note
that $P(\mathbf{w}|\mathbf{c})$ actually is the \noun{Gauss} probability
density) \begin{eqnarray*}
\log_{2}P(\mathbf{w}|\mathbf{c}) & \hspace{-2mm}= & \hspace{-2mm}\sum_{\Mdeptha=1}^{\Trank}\log_{2}\frac{1}{\sqrt{2\pi}\sigma}\exp\left(-\frac{(w_{\Mdeptha}-c_{\Mdeptha})^{2}}{2\sigma^{2}}\right)\\
\\ & \hspace{-30mm}= & \hspace{-15mm}\sum_{\Mdeptha=1}^{\Trank}\left(\log_{2}\frac{1}{\sqrt{2\pi}\sigma}-\frac{(w_{\Mdeptha}-c_{\Mdeptha})^{2}}{2\sigma^{2}\cdot\ln2}\right)\\
 & \hspace{-30mm}= & \hspace{-15mm}\Trank\log_{2}\frac{1}{\sqrt{2\pi}\sigma}-\frac{1}{2\sigma^{2}\cdot\ln2}\sum_{\Mdeptha=1}^{\Trank}\left(c_{\Mdeptha}^{2}+w_{\Mdeptha}^{2}-2c_{\Mdeptha}w_{\Mdeptha}\right)\\
 & \hspace{-30mm}= & \hspace{-15mm}\underbrace{\Trank\log_{2}\frac{1}{\sqrt{2\pi}\sigma}-\frac{n+\mathbf{w}\mathbf{w}^{T}}{2\sigma^{2}\cdot\ln2}}_{\KIb}+\mathbf{cw}^{T}\cdot\underbrace{\frac{1}{\sigma^{2}\cdot\ln2}}_{\KII}\,.\end{eqnarray*}

\end{itemize}
In either case we can thus express the conditional uncertainty as
\[
H(\mathbf{c}|\mathbf{w})=\underbrace{\left(\KIa-\KIb\right)}_{\KI}-\KII\cdot\mathbf{cw}^{T}\:,\]
I.e., the uncertainty is linearly related to the correlation.

\subsection{Calculating the Actual Distribution\label{sub:Calculate-Actual-Distribution}}

For a trellis of rank $\Trank$ and $g_{i}(e)\in\{\pm1\}$, which
is the case for hard decision decoding, the domain of the distributions,
i.e., the values that $f(\PP)$ can take, is $\sD=\{-\Trank,-\Trank+2,\dots,\Trank-2,\Trank\}$
with cardinality $\left|\sD\right|=\Trank+1$. In this case the distributions
can be directly implemented as vectors of length $\Trank+1$. A shift
$\boxplus$ of the domain is simply a shift of the vector contents,
and the correlation operation $*$ is discrete.

In case of soft decision, the domain needs to be quantized. For \noun{Gauss}ian
distributions, an efficient way for uniform mid-tread quantization
 is to carry along the mean value $\mu$ of the distribution and
to arrange the partitions equally to both sides of it, storing the
partition contents in vectors $\Dd$. When extending a path $\PP$
by an edge $e\in\sE_{i-1,i}$ in the forward/backward recursion (\emph{lengthening}),
the domain of $f(\PP)$ is shifted by $g_{i}(e)$, i.e., $g_{i}(e)$
is added to the $\mu$. However, when \emph{joining} paths in a vertex,
the mean values of the incoming path distributions do usually not
coincide. Hence a new mean value $\Dmunew$ has to be determined and
the partition contents need to be distributed. 

Let the vectors $\Dd$ be of length $(2\Nslots+1)$, each element
corresponding to a partition of width $\Swidth\,$. The partitions
are indexed by $\Didx\in\{-N,-N+1,\dots,N\}$, where $\Didx=0$ denotes
the center partition around the mean value. The a mean value $\Dmunew$
is the weighted sum of the mean values $\Dmuin$ of the involved distributions
in vectors $\Dd_{\mathrm{in}}$. E.g., for the forward recursion,
 \begin{multline*}
\Dmunew=\alpha^{\mu}(v):=\sum_{e:\mathrm{fin}(e)=v}\underbrace{\left(\alpha^{\mu}\left(\mathrm{init}(e)\right)+g_{i}(e)\right)}_{\Dmuin}\\
\cdot\underbrace{\frac{\lambda(e)\cdot\vFlow{A}{\mathrm{init}(e)}}{\sum_{e':\mathrm{fin}(e')=v}\lambda(e')\cdot\vFlow{A}{\mathrm{init}(e')}}}_{\mathrm{relative\, weight\, of\, edge\,}e},\end{multline*}
 with $\alpha^{\fD}(A)=0$, where $\Dmuin$ is the mean value of the
distribution $\alpha^{\fD}\left(\mathrm{init}(e)\right)$ after lengthening
by $g_{i}(e)$, and $\alpha^{\mu}(v)$ is the mean of the forward
distribution $\alpha^{\fD}(v)$. The final distribution vector $\Dd_{\mathrm{new}}=\alpha^{\mathbf{d}}(v)$
is the weighted sum of the vectors $\Dd_{\mathrm{out}}$ which are
calculated by distributing the content of the vectors $\Dd_{\mathrm{in}}=\alpha^{\mathbf{d}}\left(\mathrm{init}(e)\right)$
according to the new partition margins with $\DmuDelta=\mu_{\mathrm{new}}-\mu_{\mathrm{in}}$
as follows (cf. Figure \ref{fig:Distribution-Assignment}).%
\begin{figure*}
\begin{centering}
\includegraphics[scale=0.4]{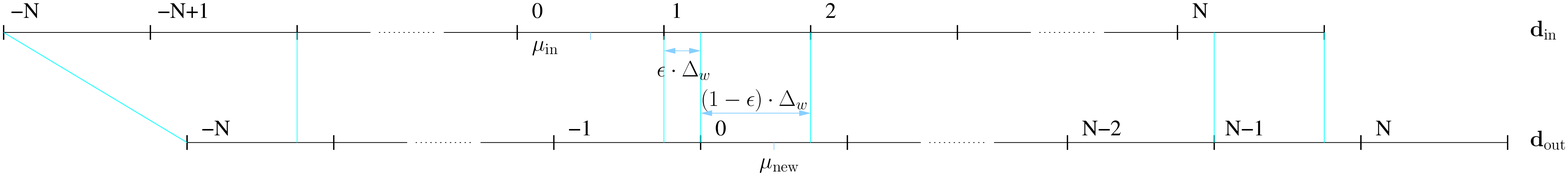} 
\par\end{centering}

\caption{\label{fig:Distribution-Assignment}Assignment of Partition Contents
of a Quantized Distribution}

\end{figure*}

\begin{itemize}
\item $\Dfraclow=\frac{\DmuDelta}{\Swidth}-\left\lfloor \frac{\DmuDelta}{\Swidth}\right\rfloor \in[0,1)$
\item $\Dd_{\mathrm{out}}=\mathbf{0}$ (all-zero vector)
\item $\Del{\mathrm{out}}{-\Nslots}={\displaystyle \sum_{\Didx=-\Nslots}^{-\Nslots+\left\lfloor \frac{\DmuDelta}{\Swidth}\right\rfloor }}\Del{\mathrm{in}}{\Didx}$
\item $\Del{\mathrm{out}}{\Nslots}=\sum_{\Didx=\Nslots+\left\lfloor \frac{\DmuDelta}{\Swidth}\right\rfloor +1}^{\Nslots}\Del{\mathrm{in}}{\Didx}$
\item for {\footnotesize $\max\left(-\Nslots,-\Nslots+\left\lfloor \frac{\DmuDelta}{\Swidth}\right\rfloor +1\right)\leq\Didx\leq\min\left(\Nslots,\Nslots+\left\lfloor \frac{\DmuDelta}{\Swidth}\right\rfloor \right)$}

\begin{itemize}
\item $\Del{\mathrm{out}}{\Didx-\left\lfloor \frac{\DmuDelta}{\Swidth}\right\rfloor -1}\pe\Dfraclow\cdot\Del{\mathrm{in}}{\Didx}$\vspace{0.1mm}

\item $\Del{\mathrm{out}}{\Didx-\left\lfloor \frac{\DmuDelta}{\Swidth}\right\rfloor }\pe(1-\Dfraclow)\cdot\Del{\mathrm{in}}{\Didx}$
\end{itemize}
\end{itemize}
where $a\pe b$ denotes the addition of $b$ to $a$, i.e., $a=a+b$.
The forward distribution vector $\alpha^{\mathbf{d}}(A)$ is initialized
The backward distribution is computed analogously.

With the two procedures of lengthening and joining the mean value
of the symbol distribution can be calculated by \begin{multline*}
\Omega_{i}^{\mu}(\MsymS,\TT)=\sum_{{e\in E_{i-1,i}:\atop \fs(e)=\MsymS}}\underbrace{\left(\alpha^{\mu}\left(\mathrm{init}(e)\right)+g_{i}(e)+\beta^{\mu}\left(\mathrm{fin}(e)\right)\right)}_{\mu_{\mathrm{in}}}\\
\cdot\underbrace{\frac{\vFlow{A}{\mathrm{init}(e)}\cdot\lambda(e)\cdot\vFlow{\mathrm{fin}(e)}{B}}{\sum_{e'\in E_{i-1,i}:c(e')=\MsymS}\vFlow{A}{\mathrm{init}(e')}\cdot\lambda(e')\cdot\vFlow{\mathrm{fin}(e')}{B}}}_{\mathrm{relative\, weight\, of\, edge\,}e},\end{multline*}
 and the discrete symbol distribution vector $\mathbf{d}_{\mathrm{out}}=\Omega_{i}^{\mathbf{d}}(\MsymS,\TT)$
is obtained by convolving the forward and backward distribution vectors
$\alpha^{\mathbf{d}}\left(\mathrm{init}(e)\right)$ and $\beta^{\mathbf{d}}\left(\mathrm{fin}(e)\right)$
for each edge $e\in\sE_{i-1,i}:c(e)=\MsymS$,\[
\mathbf{d}_{\mathrm{in}}=\alpha^{\mathbf{d}}\left(\mathrm{init}(e)\right)*\beta^{\mathbf{d}}\left(\mathrm{fin}(e)\right),\]
 followed by a weighted re-distribution of the vector contents of
the $\mathbf{d}_{\mathrm{in}}$ to $\mathbf{d}_{\mathrm{out}}$.

\subsection{\label{sub:Appendix-Semi-Ring}Generalization to Calculations on
a Semi-ring}

In the main part of this paper, the computation of moments in the
trellis is introduced for real numbers. However, the algorithm is
valid for the more general algebraic structure of commutative semi-rings.
The $0$-th forward moment then results in the Viterbi algorithm on
semi-rings.

Let the $\lambda$-label and the $c$-label come from an algebraic
set $\sS$ which is closed under the two binary operations $\oPlus$
and $\oDot$, called addition and multiplication, which satisfy the
following axioms:

\begin{itemize}
\item The operation $\oDot$ is associative and commutative, and there is
an identity element $\vSIdent$ such that $s\oDot\vSIdent=\vSIdent\oDot s=s$
for all $s\in\sS$, making $(\sS,\oDot)$ a \emph{commutative monoid}.
\item The operation $\oPlus$ is associative and commutative, and there
is an identity element $\vSNull$ such that $s\oPlus\vSNull=\vSNull\oPlus s=s$
for all $s\in\sS$, making $(\sS,\oPlus)$ a \emph{commutative monoid}.
\item The distributive law $(x\oPlus y)\oDot z=(x\oDot z)\oPlus(y\oDot z)$,
for all triples $(x,y,z)$ from $\sS$.
\item The identity element $\vSNull$ of the addition annihilates $\sS$,
i.e., $\vSNull\oDot s=s\oDot\vSNull=\vSNull$ for all $s\in\sS$.
\end{itemize}
The triple $(\sS,\oDot,\oPlus)$ is called a \emph{commutative semiring}.

Let $a,b\in(\sS,\oDot,\oPlus)$ be elements of such a commutative
semiring. We define the following notation:\begin{eqnarray*}
a^{m} & := & \begin{cases}
\underbrace{a\oDot a\oDot\dots\oDot a}_{m} & m\in\mathbb{N}\\
\vSIdent & m=0\end{cases}\\
na & := & \begin{cases}
\underbrace{a\oPlus a\oPlus\dots\oPlus a}_{n}={\displaystyle \oSum_{i=1}^{n}}\, a & n\in\mathbb{N}\\
\vSNull & n=0\end{cases}\end{eqnarray*}
with $n\, a\oDot b=n\,(a\oDot b)$ and $\mathbb{N}$ being the set
of natural numbers. Then the binomial theorem can be written as \[
(a\oPlus b)^{m}=\oSum_{\Mrunidx=0}^{\Midx}{\Midx \choose \Mrunidx}a^{\Mrunidx}\oDot b^{\Midx-\Mrunidx},\;\Midx,\Mrunidx\in\mathbb{N}_{0},\; a,b\in(\sS,\oDot,\oPlus)\]
 with the binomial coefficient ${\Midx \choose \Mrunidx}\in\mathbb{N}_{0}=\left\{ \mathbb{N}\cup0\right\} $.
In analogy to Definition \ref{def:forward_trellis_moment} and Theorem
\ref{thm:forward-moment} we can now define the forward numerator
and its calculation on a semi-ring.

\begin{definitn}
We define the $m$-th forward numerator of a function $f\in(\sS,\oDot,\oPlus)$
at vertex $v$ of a trellis $T$ as\begin{equation}
\alpha^{(\Midx)}(v):=\oSum_{\PP:A\rightarrow v}\lambda(\PP)\oDot\left(f(\PP)\right)^{\Midx}\label{eq:alpha_definition_semiring}\end{equation}
with initial values\[
\alpha^{(\Midx)}(A):=\left\{ \begin{array}{ccc}
\vSIdent & : & \Midx=0\\
\vSNull & : & \Midx>0\end{array}\right..\]

\end{definitn}
\begin{thm}
The $m$-th \emph{forward moment} $\alpha^{(\Midx)}(v)$ of a vertex
$v\in V_{i}$ on depth $i$ can be recursively calculated on a trellis
$T$ and a commutative semiring $(\sS,\oDot,\oPlus)$ by\begin{equation}
\alpha^{(\Midx)}(v)=\hspace{-3mm}\oSum_{e:\mathrm{fin}(e)=v}\hspace{-3mm}\lambda(e)\oDot\oSum_{\Mrunidx=0}^{\Midx}{\Midx \choose \Mrunidx}\left(g_{i}(e)\right)^{\Mrunidx}\oDot\alpha^{(\Midx-\Mrunidx)}\left(\mathrm{init}(e)\right)\label{eq:alpha_recursive_formula_semiring}\end{equation}
 for all functions $f(\PP:A\rightarrow v)$ and $g_{j}$, $j=1,\dots,i$,
which fulfill\begin{equation}
f(\PP)=f(e_{1}e_{2}\cdots e_{i})=g_{1}(e_{1})\oPlus g_{2}(e_{2})\oPlus\dots\oPlus g_{i}(e_{i}).\label{eq:property_of_f_semiring}\end{equation}

\end{thm}
\begin{proof}
The proof is by induction on $\mathrm{depth}(v)$. For $\mathrm{depth}(v)=1$
the algorithm computes\begin{eqnarray*}
\alpha^{(\Midx)}(v) & = & \oSum_{e:\mathrm{fin}(e)=v}\lambda(e)\oDot\left(1\:\left(g_{1}(e)\right)^{\Midx}\oDot\vSIdent\right)\\
 & = & \oSum_{e:\mathrm{fin}(e)=v}\lambda(e)\oDot\left(g_{1}(e)\right)^{\Midx}\end{eqnarray*}
which is, as required, the sum of the labels on all edges $e$ joining
$A$ to $v$, weighted by $\left(g_{1}(e)\right)^{\Midx}$. For a
vertex $v$ at depth $i+1$ the value assigned to $\alpha^{(\Midx)}(v)$
is by the induction hypothesis \begin{multline*}
\alpha^{(\Midx)}(v)=\oSum_{e:\mathrm{fin}(e)=v}\lambda(e)\oDot\oSum_{\Mrunidx=0}^{\Midx}{\Midx \choose \Mrunidx}\left(g_{i}(e)\right)^{\Mrunidx}\\
\oDot\oSum_{\PP:A\rightarrow\mathrm{init}e}\lambda(\PP)\oDot\left(f(\PP)\right)^{\Midx-\Mrunidx}.\end{multline*}
 Using the axioms%
\footnote{- $A\oDot{\displaystyle \oSumFoot_{i}}B_{i}={\displaystyle \oSumFoot_{i}}A\oDot B_{i}$
requires distributive law (factor into sum)

~- ${\displaystyle \oSumFoot_{i}}{\displaystyle \oSumFoot_{j}}A_{ij}={\displaystyle \oSumFoot_{j}}{\displaystyle \oSumFoot_{i}}A_{ij}$
requires associativity and commutativity \textcolor{white}{MA$\hspace{-3mm}$RK}
of $\oPlus$ (change order of sums)

~- $A\oDot B=B\oDot A$ requires commutativity of $\oDot$ (change
order of factors)%
} of the commutative semiring $(\sS,\oDot,\oPlus)$ we have \begin{multline*}
\alpha^{(\Midx)}(v)=\oSum_{e:\mathrm{fin}(e)=v}\oSum_{\PP:A\rightarrow\mathrm{init}e}\lambda(e)\oDot\lambda(\PP)\\
\oDot\oSum_{\Mrunidx=0}^{\Midx}{\Midx \choose \Mrunidx}\left(g_{i}(e)\right)^{\Mrunidx}\oDot\left(f(\PP)\right)^{\Midx-\Mrunidx}.\end{multline*}
 Applying Equation (\ref{eq:property_of_f_semiring}) and the binomial
theorem we obtain \[
\alpha^{(\Midx)}(v)=\oSum_{e:\mathrm{fin}(e)=v}\oSum_{\PP:A\rightarrow\mathrm{init}e}\lambda(\PP e)\oDot\left(f(\PP)\oPlus g_{i}(e)\right)^{\Midx}.\]
 But every path from $A$ to $v$ must be of the form $\PP e$, where
$\PP$ is a path from $A$ to a vertex $u$ with $\mathrm{depth}(u)=i$,
$\mathrm{init}(e)=u$ and $\mathrm{fin}(e)=v$. Hence, $\alpha^{(\Midx)}(v)$
is correctly calculated by the theorem.
\end{proof}
\begin{remrk}
Note that the complexity considerations in Theorems \ref{thm:MomentCalculationComplexity}
and \ref{thm:SymbolMomentComplexity} transfer to the calculation
on semi-rings. However, the terminology of {}``addition'' and {}``multiplication''
then refers to the operations $\oPlus$ and $\oDot$.
\end{remrk}

\bibliographystyle{unsrt}
\bibliography{TrellisComputations}

\end{document}